\documentclass[english,letterpaper, 10 pt]{ieeeconf}
\usepackage[T1]{fontenc}
\usepackage[latin9]{inputenc}
\usepackage{amsmath}
\usepackage{amssymb}
\usepackage{graphicx}

\makeatletter
\usepackage{cite}
  \usepackage{times}   
\IEEEoverridecommandlockouts 
\newtheorem{theorem}{Theorem}[section]
\newtheorem{definition}[theorem]{Definition}
\newtheorem{remark}[theorem]{Remark}
\newtheorem{corollary}[theorem]{Corollary}
\newtheorem{lemma}[theorem]{Lemma}
\newtheorem{proposition}[theorem]{Proposition}
\newtheorem{aplemma}{Lemma}[section]

\renewcommand{\QED}{\QEDopen}
\usepackage[all]{xy}

\usepackage{balance}

\usepackage{tikz}
\usetikzlibrary{positioning, arrows, matrix}

\usepackage{psfrag}

\usepackage{dsfont}

\newtheorem{assumption}{Assumption}[section]

\usepackage{url}

\makeatother

\usepackage{babel}
\begin{document}

\title{\LARGE \bf Event-Triggered Control over Unreliable Networks \\
Subject to Jamming Attacks}

\author{Ahmet Cetinkaya, Hideaki Ishii, and Tomohisa Hayakawa \thanks{A. Cetinkaya and T. Hayakawa are with the Department of Mechanical and Environmental Informatics, Tokyo Institute of Technology, Tokyo 152-8552, Japan. {\tt\small{ahmet@dsl.mei.titech.ac.jp, hayakawa@mei.titech.ac.jp}}}
\thanks{H. Ishii is with the Department of Computational Intelligence and Systems Science, Tokyo Insitute of Technology, Yokohama, 226-8502, Japan. {\tt\small{ishii@dis.titech.ac.jp}}}
\thanks{This work was supported in part by Japan Science and Technology Agency under the CREST program.}}

\maketitle
\begin{abstract}Event-triggered networked control of a linear dynamical
system is investigated. Specifically, the dynamical system and the
controller are assumed to be connected through a communication channel.
State and control input information packets between the system and
the controller are attempted to be exchanged over the network only
at time instants when certain triggering conditions are satisfied.
We provide a probabilistic characterization for the link failures
which allows us to model random packet losses due to unreliability
in transmissions as well as those caused by malicious jamming attacks.
We obtain conditions for the almost sure stability of the closed-loop
system, and we illustrate the efficacy of our approach with a numerical
example.\end{abstract}

\section{Introduction}

One of the main challenges in networked control systems is that communication
between plant and controller may not always be reliable. State measurement
and control input packets may fail to be transmitted at times due
to network congestion or errors in communication. In the literature,
unreliability of a network is often characterized through random models
for packet loss events \cite{hespanha2007}. For instance, in \cite{ishii2009,Lemmon:2011:ASS:1967701.1967744},
Bernoulli processes are used for modeling packet losses in a network.
Furthermore, in \cite{gupta2009,okano2014}, packet loss events are
characterized in a more general way by employing Markov chains. In
these studies, a variety of control methods are proposed to ensure
stability of networked control systems that face random packet losses.

More recently, cyber security has become a critical issue in networked
control systems since the channels are nowadays connected via the
Internet or wireless communications \cite{fawzi2014,wholejournal2015}.
Here, in addition to random losses, we consider communication effects
due to jamming attacks initiated by malicious agents. Such attacks
may block the communication link and effectively prevent transmission
of packets between the plant and the controller. In a few recent works
\cite{amin2009,HSF-SM:ct-13,depersis2014,liu2014stochastic}, networked
control problems under malicious jamming attacks were investigated.

In this paper we explore feedback control of a discrete-time linear
system over a network that is subject to both random packet losses
due to unreliability of the communication channel and jamming attacks
coming from an intelligent attacker. We employ an event-triggered
control framework, where the plant and the controller attempt to exchange
state and control input information packets only at times that correspond
to event-triggering instants. Following the approach in \cite{velasco2009,heemels2013periodic},
we utilize Lyapunov-like functions to characterize the triggering
conditions. The triggering conditions that we propose in this paper
ensure that the value of a Lyapunov-like function of the state stays
within certain limits. Packet exchanges are attempted only before
the value of the Lyapunov-like function is predicted to exceed a certain
level. In a successful packet exchange scenario, state measurements
are sent from the plant to the controller, which computes a control
input and sends it back to the plant. However, state measurement or
control input packets may fail to be transmitted due to random packet
losses and jamming attacks. We model random losses using a binary-valued
\emph{time-inhomogeneous} Markov chain. To characterize jamming attacks,
we follow the approach of \cite{depersis2014}. Rather than specifying
predetermined patterns or distributions for the occurrences of jamming
attacks, we allow jamming attacks to happen arbitrarily as long as
the total number of packet exchange attempts that face jamming attacks
are almost surely bounded by a certain ratio of the number of total
packet exchange attempts. 

We consider both the case where random packet losses and jamming attacks
are independent and the case where the attacker may use information
of past random packet losses in generating a jamming attack strategy.
The main theoretical challenge in dealing with both of these cases
stems from the fact that random losses and jamming attacks are of
different nature and hence have different models. By utilizing a tail
probability inequality for the sum of processes that represent random
losses and jamming attacks, we show that a probabilistic characterization
for the evolution of the total number of packet exchange failures
allows us to deal with both cases. Based on this characterization,
we obtain conditions for almost sure asymptotic stability of the closed-loop
event-triggered networked control system. Furthermore, we present
a numerical method for finding stabilizing feedback gains as well
as parameters for the event-triggering mechanism. 

The paper is organized as follows. In Section~\ref{sec:Event-Triggered-Control-over},
we describe the event-triggered networked control system under random
and jamming-related packet losses. We provide sufficient conditions
for almost sure asymptotic stability of the closed-loop control system
in Section~\ref{sec:Conditions-for-Almost-Sure} and present an illustrative
numerical example in Section~\ref{sec:Numerical-Example}. Finally,
in Section~\ref{sec:Conclusion}, we conclude the paper. 

We use a fairly standard notation in the paper. Specifically, we denote
positive and nonnegative integers by $\mathbb{N}$ and $\mathbb{N}_{0}$,
respectively. We write $\mathbb{R}$ for the set of real numbers,
$\mathbb{R}^{n}$ for the set of $n\times1$ real column vectors,
and $\mathbb{R}^{n\times m}$ for the set of $n\times m$ real matrices.
Moreover, $(\cdot)^{\mathrm{T}}$ denotes transpose, $\|\cdot\|$
denotes the Euclidean vector norm, and $\left\lfloor \cdot\right\rfloor $
denotes the largest integer that is less than or equal to its real
argument. The notation $\mathrm{\mathbb{P}}[\cdot]$ denotes the probability
on a probability space $(\Omega,\mathcal{F},\mathbb{P})$ with filtration
$\{\mathcal{F}_{i}\}_{i\in\mathbb{N}_{0}}$.

\section{Event-Triggered Networked Control \label{sec:Event-Triggered-Control-over}}

In this section we provide the mathematical model for the event-triggered
networked control system. We then present a characterization of a
network that faces random packet losses and packet losses caused by
jamming attacks.

\subsection{Event-Triggered Control System}

Consider the linear dynamical system 
\begin{align}
x(t+1) & =Ax(t)+Bu(t),\quad x(0)=x_{0},\quad t\in\mathbb{N}_{0},\label{eq:system}
\end{align}
 where $x(t)\in\mathbb{R}^{n}$ and $u(t)\in\mathbb{R}^{m}$ denote
the state and the control input; furthermore, $A\in\mathbb{R}^{n\times n}$
and $B^{n\times m}$ are the state and input matrices, respectively. 

In this paper, we use the event-triggering framework (see \cite{heemels2013periodic}
and the references therein), where the control input is only updated
when a certain triggering condition is satisfied. The triggering condition
is checked at each time step at the plant side. 

In the networked operation setting, the plant and the controller are
connected through a communication channel and attempt to exchange
information packets at times corresponding to event-triggering instants.
We consider the case where packets are transmitted without delay,
but they may get lost. In a successful packet exchange scenario, at
a certain time instant, measured plant states are transmitted to the
controller, which generates a control input based on the received
state information and sends a packet containing the control input
information to the plant. The transmitted control input is applied
at the plant side. In the case of an unsuccessful packet exchange
attempt, either the measured state packet or the control input packet
may get dropped, and in such cases control input at the plant side
is set to $0$. 

We use $\tau_{i}\in\mathbb{N}_{0},i\in\mathbb{N}_{0}$, to denote
the time instants at which packet exchanges between the plant and
the controller are attempted. To characterize these time instants
we utilize a quadratic Lyapunov-like function $V\colon\mathbb{R}^{n}\to[0,\infty)$
given by $V(x)\triangleq x^{\mathrm{T}}Px$, where $P>0$. Letting
$\tau_{0}=0$, we describe $\tau_{i}$, $i\in\mathbb{N}$, and control
input $u(t)$ applied to the plant by 
\begin{align}
\tau_{i+1} & \triangleq\min\Big\{ t\in\{\tau_{i}+1,\tau_{i}+2,\ldots\}\colon t\geq\tau_{i}+\theta\nonumber \\
 & \quad\quad\quad\mathrm{or}\,\,\,V(Ax(t)+Bu(\tau_{i}))>\beta V(x(\tau_{i}))\Big\},\label{eq:attemptedpacketexchangetimes}\\
u(t) & \triangleq\left(1-l(i)\right)Kx(\tau_{i}),\,t\in\{\tau_{i},\ldots,\tau_{i+1}-1\},\label{eq:controlinputatplantside}
\end{align}
 for $i\in\mathbb{N}_{0}$, where $\beta\in(0,1)$, $\theta\in\mathbb{N}$,
and $\{l(i)\in\{0,1\}\}_{i\in\mathbb{N}_{0}}$ is a binary-valued
process that characterizes success or failure of packet exchange attempts.
When $l(i)=0$, packet exchange attempt at time $\tau_{i}$ is successful
and the piecewise-constant control input at the plant side is set
to $u(\tau_{i})=Kx(\tau_{i})$, where $K\in\mathbb{R}^{m\times n}$
denotes the feedback gain. On the other hand, the case $l(i)=1$ indicates
that either the packet sent from the plant or the packet sent from
the controller is lost at time $\tau_{i}$. Again, in such cases,
control input at the plant side is set to $0$. 

The triggering condition (\ref{eq:attemptedpacketexchangetimes})
involves two parts. The part characterized by $V(Ax(t)+Bu(\tau_{i}))>\beta V(x(\tau_{i}))$
ensures that after a successful packet exchange attempt at $\tau_{i}$,
the value of the Lyapunov-like function $V(\cdot)$ stays below the
level $\beta V(x(\tau_{i}))$ until the next packet exchange attempt.
Furthermore, the triggering condition $t\geq\tau_{i}+\theta$ ensures
that two consecutive packet exchange attempt instants are at most
$\theta$ steps apart, that is, $\tau_{i+1}-\tau_{i}\leq\theta$,
$i\in\mathbb{N}_{0}$. Although the specific value of $\theta$ does
not affect the results developed below, the boundedness of packet
exchange attempt intervals guarantees that $\tau_{i}$ (and hence
$V(x(\tau_{i}))$) are well-defined for each $i\in\mathbb{N}$. In
practice, the value of $\theta$ can be selected considering how frequent
the plant state is desired to be monitored by the controller side. 

\begin{figure}
\centering  \includegraphics[width=0.85\columnwidth]{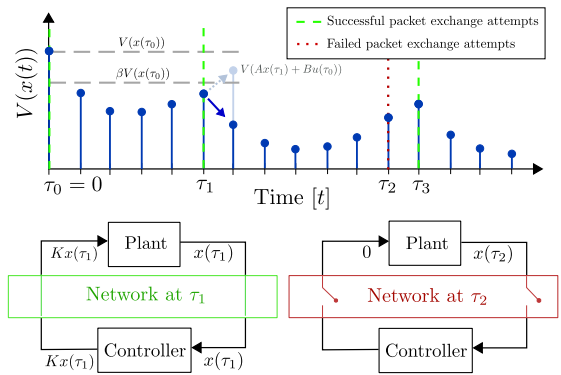} 

\caption{Networked control system operation}
 \label{operation}
\end{figure}

Operation of the event-triggered networked control system is illustrated
in Fig.~\ref{operation}. The triggering condition (\ref{eq:attemptedpacketexchangetimes})
is checked at the plant side at each step $t\in\mathbb{N}_{0}$. At
times $t=\tau_{i}$, $i\in\mathbb{N}$, the triggering condition is
satisfied and packet exchanges between the plant and the controller
are attempted. For this example, packet exchange is attempted at time
$t=\tau_{1}$, since $V(Ax(t)+Bu(\tau_{0}))>\beta V(x(\tau_{0}))$.
At this time instant, the plant and the controller successfully exchange
state and control input packets over the network, and as a result,
control input on the plant side is updated to $Kx(\tau_{1})$. Note
that packet exchange  attempts are not always successful, and may
fail due to loss of packets in the network. For instance, the packet
exchange attempt at time $\tau_{2}$ fails for the case of Fig. \ref{operation}.
In this case, the control input at the plant side is set to $0$ at
time $\tau_{2}$, which results in an unstable behavior. A packet
exchange is attempted again at the very next time step $\tau_{3}$,
since the triggering condition is also satisfied at that instant.

\begin{remark}The event-triggering framework we describe above requires
a plant-side mechanism for checking the triggering condition (\ref{eq:attemptedpacketexchangetimes})
at each time step. If the overall practical setup of the process does
not allow placing such a mechanism at the plant side, following the
self-triggering control approach described in \cite{heemels2012cdc},
we can use a decision mechanism at the controller side. In this case,
packet exchange times are decided at the controller side based on
the state information obtained at previously successful packet exchange
attempts. \end{remark}

\subsection{Characterization of a Network with Random Packet Losses and Packet
Losses Caused by Jamming Attacks}

Packet transmission failures in a network may have different reasons.
Packet losses caused by network congestion may be accurately described
using stochastic models \cite{altman2005}. However, only stochastic
models would not be enough to characterize packet losses if the communication
channel is subject to jamming attacks of a malicious agent. In what
follows we characterize the effects of certain stochastic and jamming-related
packet loss models in a unified manner by investigating dynamical
evolution of the total number of packet exchange failures. 

First, we define a nonnegative integer-valued process $\{L(k)\in\mathbb{N}_{0}\}_{k\in\mathbb{N}}$
by 
\begin{align}
L(k) & \triangleq\sum_{i=0}^{k-1}l(i),\quad k\in\mathbb{N}.\label{eq:bigldefinition}
\end{align}
 Note that $L(k)$ denotes the total number of \emph{failed} packet
exchange attempts during the time interval $[0,\tau_{k-1}]$. 

\begin{assumption} \label{MainAssumption} There exist scalars $\rho\in[0,1]$,
$\gamma_{k}\in[0,\infty)$, $k\in\mathbb{N}$, such that 
\begin{align}
\mathbb{P}[L(k)>\rho k] & \leq\gamma_{k},\quad k\in\mathbb{N},\label{eq:lcond1}\\
\sum_{k\in\mathbb{N}}\gamma_{k} & <\infty.\label{eq:lcond2}
\end{align}

\end{assumption} \vskip 8pt

Note that conditions (\ref{eq:lcond1}) and (\ref{eq:lcond2}) provide
a probabilistic characterization of the evolution of the total number
of packet exchange failures through scalars $\rho\in[0,1]$, $\gamma_{k}\in[0,\infty)$,
$k\in\mathbb{N}$. A closely related characterization for packet dropouts
in a communication link is presented in \cite{Lemmon:2011:ASS:1967701.1967744};
the scalar $\rho$ in (\ref{eq:lcond1}) corresponds to the notion
\emph{dropout rate} discussed there. 

The following result is a direct consequence of Borel-Cantelli lemma
(see \cite{klenke2008}) and it shows that under Assumption~\ref{MainAssumption},
the long run average of the total number of failed packet exchanges
is upper bounded by $\rho$. 

\vskip 3pt

\begin{lemma} \label{key-lemma1} If there exist scalars $\rho\in[0,1]$,
$\gamma_{k}\in[0,\infty)$, $k\in\mathbb{N}$, such that (\ref{eq:lcond1}),
(\ref{eq:lcond2}) hold, then $\limsup_{k\to\infty}\frac{L(k)}{k}\leq\rho$,
almost surely. \end{lemma}

\vskip 3pt

It is important to note that for any packet loss model, Assumption~\ref{MainAssumption}
is trivially satisfied with $\rho=1$, and $\gamma_{k}=0,k\in\mathbb{N}$,
since $L(k)\leq k$. On the other hand, as illustrated in the following,
for certain random and jamming-related packet loss models, $\rho$
can be obtained to be strictly smaller than $1$. 

\vskip 5pt

\subsubsection{Random Losses\label{RemarkRandomPacketLoss}}

To characterize random packet losses in the communication channel,
we utilize time-inhomogeneous Markov chains. Specifically, let $\{l_{\mathrm{R}}(i)\in\{0,1\}\}_{i\in\mathbb{N}_{0}}$
be an $\mathcal{F}_{i}$-adapted time-inhomogeneous Markov chain characterized
by initial distributions $\vartheta_{q}\in[0,1]$, $q\in\{0,1\}$,
and time-varying transition probabilities $p_{q,r}\colon\mathbb{N}_{0}\to[0,1]$,
$q,r\in\{0,1\}$, such that 
\begin{align*}
\mathbb{P}[l_{\mathrm{R}}(0)=q] & =\vartheta_{q},\,\,q\in\{0,1\},\\
\mathbb{P}[l_{\mathrm{R}}(i+1)=r|l_{\mathrm{R}}(i)=q] & =p_{q,r}(i),\,q,r\in\{0,1\},\,i\in\mathbb{N}_{0}.
\end{align*}

The state $l_{\mathrm{R}}(i)=1$ indicates that the network faces
random packet losses at time $\tau_{i}$, and hence the packet exchange
attempt at $\tau_{i}$ results in failure. Note that in this characterization,
the event that a packet exchange attempt fails depends on the states
of previous packet exchange attempts. Furthermore, transition probabilities
between success ($l_{\mathrm{R}}(i)=0$) and failure ($l_{\mathrm{R}}(i)=1$)
states of packet exchange attempts are time-dependent. For instance,
the probability of packet exchange failure at time $\tau_{i+1}$ is
given by $p_{l_{\mathrm{R}}(i),1}(i)$. 

It is important to note that the time-inhomogeneous Markov chain characterization
of random packet losses generalizes the Bernoulli and \emph{time-homogeneous}
Markov chain models that are often used in the literature. \vskip 5pt

\subsubsection{Jamming Attacks\label{RemarkJammingAttacks} }

For jamming attacks coming from an intelligent attacker, a model capturing
the attack strategy of a malicious agent has been proposed in \cite{depersis2014}.
In that study, the sum of the length of attack durations is assumed
to be bounded by a certain ratio of total time.

We follow the approach of \cite{depersis2014} for modeling packet
exchange failures due to a jamming attack. Specifically, let $\{l_{\mathrm{J}}(i)\in\{0,1\}\}_{i\in\mathbb{N}_{0}}$
denote the state of jamming attacks. The state $l_{\mathrm{J}}(i)=1$
indicates that the network is subject to a jamming attack at time
$\tau_{i}$. We consider the case where the number of packet exchange
attempts that face jamming attacks are upper bounded almost surely
by a certain ratio of the total number of packet exchange attempts,
that is, $\{l_{\mathrm{J}}(i)\in\{0,1\}\}_{i\in\mathbb{N}_{0}}$ satisfies
\begin{align}
\mathbb{P}\big[\sum_{i=0}^{k-1}l_{\mathrm{J}}(i)\leq\kappa+\frac{k}{\tau}\big]=1,\quad k\in\mathbb{N},\label{eq:jammingcondition}
\end{align}
where $\kappa\ge0$ and $\tau>1$. In this characterization, among
$k$ packet exchange attempts, at most $\kappa+\frac{k}{\tau}$ of
them are affected by jamming attacks. The ratio $\frac{1}{\tau}$
corresponds to the notion \emph{jamming rate} discussed in \cite{anantharamu2011}.
Note that when $\kappa=0$, (\ref{eq:jammingcondition}) implies $l_{\mathrm{J}}(i)=0$,
$i\in\{0,\ldots,\lfloor\tau\rfloor\}$, almost surely. Scenarios that
involve possible jamming attacks during the first few packet exchange
attempts can be modeled by setting $\kappa>0$. Note that the characterization
in (\ref{eq:jammingcondition}) does not require $l_{\mathrm{J}}(i)$,
$i\in\mathbb{N}_{0}$, to follow a particular distribution. In fact,
$l_{\mathrm{J}}(\cdot)$ may be generated in a deterministic fashion,
or it may involve randomness. \vskip 5pt

\subsubsection{Combination of Random and Jamming-Related Losses}

In order to model the case where the network is subject to both random
losses and malicious jamming attacks, we define $\{l(i)\in\{0,1\}\}_{i\in\mathbb{N}_{0}}$
by 
\begin{align}
l(i) & =\begin{cases}
1,\quad & l_{\mathrm{R}}(i)=1\,\,\mathrm{or}\,\,l_{\mathrm{J}}(i)=1,\\
0,\quad & \mathrm{otherwise},
\end{cases}\,\,\,i\in\mathbb{N}_{0},\label{eq:ldefinitionforcombinedcase}
\end{align}
 where $\{l_{\mathrm{R}}(i)\in\{0,1\}\}_{i\in\mathbb{N}_{0}}$ is
a time-inhomogeneous Markov chain characterizing random packet losses
(see Section~\ref{RemarkRandomPacketLoss}) and $\{l_{\mathrm{J}}(i)\in\{0,1\}\}_{i\in\mathbb{N}_{0}}$
satisfying (\ref{eq:jammingcondition}) is a binary-valued process
that characterizes jamming attacks (see Section~\ref{RemarkJammingAttacks}). 

Proposition~\ref{PropositionCombinedCase} below provides a range
of values for $\rho\in(0,1)$ that satisfy Assumption~\ref{MainAssumption}
in the case that the network under consideration faces both random
packet losses and jamming attacks. In the proof of this result, we
utilize Lemma~\ref{KeyMarkovLemma} for obtaining upper bounds on
tail probabilities of sums $\sum_{i=0}^{k-1}l_{\mathrm{R}}(i)$ and
$\sum_{i=0}^{k-1}(1-l_{\mathrm{R}}(i))l_{\mathrm{J}}(i)$. 

\vskip 5pt

\begin{proposition} \label{PropositionCombinedCase} Consider the
packet loss indicator process $\{l(i)\in\{0,1\}\}_{i\in\mathbb{N}_{0}}$
given by (\ref{eq:ldefinitionforcombinedcase}) where $\{l_{\mathrm{R}}(i)\in\{0,1\}\}_{i\in\mathbb{N}_{0}}$
and $\{l_{\mathrm{J}}(i)\in\{0,1\}\}_{i\in\mathbb{N}_{0}}$ are mutually
independent. If there exist scalars $p_{0},p_{1}\in(0,1)$ such that
\begin{align}
p_{q,1}(i) & \leq p_{1},\label{eq:p1condition}\\
p_{q,0}(i) & \leq p_{0},\quad q\in\{0,1\},\quad i\in\mathbb{N}_{0},\label{eq:p0condition}\\
p_{1}+\frac{p_{0}}{\tau} & <1,
\end{align}
 hold, then for all $\rho\in(p_{1}+\frac{p_{0}}{\tau},1)$, there
exist $\gamma_{k}\in[0,\infty)$, $k\in\mathbb{N}$, that satisfy
(\ref{eq:lcond1}), (\ref{eq:lcond2}).

\end{proposition} 

\vskip 5pt

\begin{proof} It follows from (\ref{eq:ldefinitionforcombinedcase})
that 
\begin{align*}
l(i) & =l_{\mathrm{R}}(i)+(1-l_{\mathrm{R}}(i))l_{\mathrm{J}}(i),\quad i\in\mathbb{N}_{0},
\end{align*}
 and hence, by (\ref{eq:bigldefinition}), 
\begin{align}
L(k) & =\sum_{i=0}^{k-1}l_{\mathrm{R}}(i)+\sum_{i=0}^{k-1}(1-l_{\mathrm{R}}(i))l_{\mathrm{J}}(i),\quad k\in\mathbb{N}.\label{eq:Linusefulform}
\end{align}
Now, let $\epsilon\triangleq\rho-p_{1}-\frac{p_{0}}{\tau}$, $\epsilon_{2}\triangleq\min\{\frac{\epsilon}{2},\frac{1-p_{0}}{2\tau}\}$,
$\epsilon_{1}\triangleq\epsilon-\epsilon_{2}$, and define $\rho_{1}\triangleq p_{1}+\epsilon_{1}$,
$\rho_{2}\triangleq\frac{p_{0}}{\tau}+\epsilon_{2}$. Note that $\rho_{1}\in(p_{1},1)$
and $\rho_{2}\in(\frac{p_{0}}{\tau},\frac{1}{\tau})$. Furthermore,
let $L_{1}(k)\triangleq\sum_{i=0}^{k-1}l_{\mathrm{R}}(i)$ and $L_{2}(k)\triangleq\sum_{i=0}^{k-1}(1-l_{\mathrm{R}}(i))l_{\mathrm{J}}(i)$.
It then follows that 
\begin{align}
\mathbb{P}[L(k)>\rho k] & =\mathbb{P}[L_{1}(k)+L_{2}(k)>\rho_{1}k+\rho_{2}k]\nonumber \\
 & =1-\mathbb{P}[L_{1}(k)+L_{2}(k)\leq\rho_{1}k+\rho_{2}k]\nonumber \\
 & \leq1-\mathbb{P}[\left\{ L_{1}(k)\leq\rho_{1}k\right\} \cap\left\{ L_{2}(k)\leq\rho_{2}k\right\} ]\nonumber \\
 & =\mathbb{P}[\left\{ L_{1}(k)>\rho_{1}k\right\} \cup\left\{ L_{2}(k)>\rho_{2}k\right\} ]\nonumber \\
 & \leq\mathbb{P}[L_{1}(k)>\rho_{1}k]+\mathbb{P}[L_{2}(k)>\rho_{2}k].\label{eq:keyprobabilityinequality}
\end{align}
In the following we obtain upper bounds for the two probability terms
on the far right-hand side of (\ref{eq:keyprobabilityinequality})
using Lemma~\ref{KeyMarkovLemma}. 

First, to apply the lemma for the term $\mathbb{P}[L_{1}(k)>\rho_{1}k]$,
let $\tilde{p}=p_{1}$, $\tilde{c}=0$, $\tilde{w}=1$, and define
processes $\{\xi(i)\in\{0,1\}\}_{i\in\mathbb{N}_{0}}$ and $\{\chi(i)\in\{0,1\}\}_{i\in\mathbb{N}_{0}}$
by setting $\xi(i)\triangleq l_{\mathrm{R}}(i)$, $\chi(i)=1$, $i\in\mathbb{N}_{0}$.
Then, the conditions in (\ref{eq:xicond}) and (\ref{eq:chicond})
are satisfied. Now, since $L_{1}(k)=\sum_{i=0}^{k-1}\xi(i)\chi(i)$,
it follows from Lemma~\ref{KeyMarkovLemma} that 
\begin{align}
\mathbb{P}[L_{1}(k)>\rho_{1}k] & \leq\gamma_{k}^{(1)},\quad k\in\mathbb{N},\label{eq:pl1result}\\
\sum_{k\in\mathbb{N}}\gamma_{k}^{(1)} & <\infty,\label{eq:gamma1result}
\end{align}
 where $\gamma_{k}^{(1)}\triangleq\phi_{1}^{-\rho_{1}k+1}\frac{\left((\phi_{1}-1)p_{1}+1\right)^{k}-1}{(\phi-1)p_{1}}$
and $\phi_{1}\triangleq\frac{\rho_{1}(1-p_{1})}{p_{1}(1-\rho_{1})}$. 

Next, we use Lemma~\ref{KeyMarkovLemma} to find an upper bound for
$\mathbb{P}[L_{2}(k)>\rho_{2}k]$. Specifically, as a consequence
of (\ref{eq:jammingcondition}), conditions (\ref{eq:xicond}), (\ref{eq:chicond})
hold with $\tilde{p}=p_{0}$, $\tilde{c}=\kappa$, and $\tilde{w}=\frac{1}{\tau}$
together with processes $\{\xi(i)\in\{0,1\}\}_{i\in\mathbb{N}_{0}}$
and $\{\chi(i)\in\{0,1\}\}_{i\in\mathbb{N}_{0}}$ defined by setting
$\xi(i)\triangleq1-l_{\mathrm{R}}(i)$, $\chi(i)=l_{\mathrm{J}}(i)$,
$i\in\mathbb{N}_{0}$. Now, since $\xi(i)\triangleq1-l_{\mathrm{R}}(i)$
and $\chi(i)=l_{\mathrm{J}}(i)$, we have $L_{2}(k)=\sum_{i=0}^{k-1}\xi(i)\chi(i)$;
and hence, Lemma~\ref{KeyMarkovLemma} implies that 
\begin{align}
\mathbb{P}[L_{2}(k)>\rho_{2}k] & \leq\gamma_{k}^{(2)},\quad k\in\mathbb{N},\label{eq:pl2result}\\
\sum_{k\in\mathbb{N}}\gamma_{k}^{(2)} & <\infty,\label{eq:gamma2result}
\end{align}
 where $\gamma_{k}^{(2)}\triangleq\phi_{2}^{-\rho_{2}k+1}\frac{\left((\phi_{2}-1)p_{0}+1\right)^{\kappa+\frac{k}{\tau}}-1}{(\phi_{2}-1)p_{0}}$
and $\phi_{2}\triangleq\frac{\tau\rho_{2}(1-p_{0})}{p_{0}(1-\tau\rho_{2})}$. 

Now, let $\gamma_{k}\triangleq\gamma_{k}^{(1)}+\gamma_{k}^{(2)}$,
$k\in\mathbb{N}$. By using (\ref{eq:keyprobabilityinequality}),
(\ref{eq:pl1result}), and (\ref{eq:pl2result}), we obtain (\ref{eq:lcond1}).
Furthermore, as a consequence of (\ref{eq:gamma1result}) and (\ref{eq:gamma2result}),
we have (\ref{eq:lcond2}) as 
\begin{align*}
\sum_{k\in\mathbb{N}}\gamma_{k} & =\sum_{k\in\mathbb{N}}\gamma_{k}^{(1)}+\sum_{k\in\mathbb{N}}\gamma_{k}^{(2)}<\infty,
\end{align*}
 which completes the proof. \end{proof} \vskip 5pt

Note that in Proposition~\ref{PropositionCombinedCase}, $\{l_{\mathrm{R}}(i)\in\{0,1\}\}_{i\in\mathbb{N}_{0}}$
and $\{l_{\mathrm{J}}(i)\in\{0,1\}\}_{i\in\mathbb{N}_{0}}$ are assumed
to be mutually independent processes. In other words, packet exchange
attempt failures due to jamming attacks are assumed to be independent
of packet exchange attempt failures due to random packet losses. This
assumption would not be satisfied in the case that the malicious jamming
attacker has information of the past random packet losses in the communication
channel and utilizes this information in the attack strategy. Proposition~\ref{PropositionDependentCombinedCase}
below deals with such cases. \vskip 5pt

\begin{proposition} \label{PropositionDependentCombinedCase} Consider
the packet loss indicator process $\{l(i)\in\{0,1\}\}_{i\in\mathbb{N}_{0}}$.
Suppose there exists a scalar $p_{1}\in(0,1)$ such that (\ref{eq:p1condition})
and 
\begin{align}
p_{1}+\frac{1}{\tau} & <1,
\end{align}
hold. Then for all $\rho\in(p_{1}+\frac{1}{\tau},1)$, there exist
$\gamma_{k}\in[0,\infty)$, $k\in\mathbb{N}$, that satisfy (\ref{eq:lcond1}),
(\ref{eq:lcond2}). 

\end{proposition} \vskip 5pt

\begin{proof}It follows from (\ref{eq:ldefinitionforcombinedcase})
that 
\begin{align*}
L(k) & \leq\sum_{i=0}^{k-1}l_{\mathrm{R}}(i)+\sum_{i=0}^{k-1}l_{\mathrm{J}}(i),\quad k\in\mathbb{N}.
\end{align*}
Let $\epsilon\triangleq\rho-p_{1}-\frac{1}{\tau}$, and define $\rho_{1}\triangleq p_{1}+\frac{\epsilon}{2}$,
$\rho_{2}\triangleq\frac{1}{\tau}+\frac{\epsilon}{2}$. Using arguments
similar to the ones used for obtaining (\ref{eq:keyprobabilityinequality})
in the proof of Proposition~\ref{PropositionCombinedCase}, we obtain
\begin{align}
 & \mathbb{P}[L(k)>\rho k]\leq\mathbb{P}[\sum_{i=0}^{k-1}l_{\mathrm{R}}(i)+\sum_{i=0}^{k-1}l_{\mathrm{J}}(i)>\rho k]\nonumber \\
 & \,\,\leq\mathbb{P}[\sum_{i=0}^{k-1}l_{\mathrm{R}}(i)>\rho_{1}k]+\mathbb{P}[\sum_{i=0}^{k-1}l_{\mathrm{J}}(i)>\rho_{2}k],\,\,k\in\mathbb{N}.\label{eq:separatedprobabilityterms}
\end{align}
The terms $\mathbb{P}[\sum_{i=0}^{k-1}l_{\mathrm{R}}(i)>\rho_{1}k]$
and $\mathbb{P}[\sum_{i=0}^{k-1}l_{\mathrm{J}}(i)>\rho_{2}k]$ in
(\ref{eq:separatedprobabilityterms}) respectively correspond to random
packet losses and packet losses due to jamming attacks. First, it
is shown in the proof of Proposition~\ref{PropositionCombinedCase}
that $\mathbb{P}[\sum_{i=0}^{k-1}l_{\mathrm{R}}(i)>\rho_{1}k]\leq\gamma_{k}^{(1)}$,
where $\gamma_{k}^{(1)}\triangleq\phi_{1}^{-\rho_{1}k+1}\frac{\left((\phi_{1}-1)p_{1}+1\right)^{k}-1}{(\phi-1)p_{1}}$
and $\phi_{1}\triangleq\frac{\rho_{1}(1-p_{1})}{p_{1}(1-\rho_{1})}$.
Moreover, using Markov's inequality we obtain 
\begin{align}
 & \mathbb{P}[\sum_{i=0}^{k-1}l_{\mathrm{J}}(i)>\rho_{2}k]\leq\mathbb{P}[\sum_{i=0}^{k-1}l_{\mathrm{J}}(i)\geq\rho_{2}k]\nonumber \\
 & \quad=\mathbb{P}[e^{\sum_{i=0}^{k-1}l_{\mathrm{J}}(i)}\geq e^{\rho_{2}k}]\leq e^{-\rho_{2}k}\mathbb{E}[e^{\sum_{i=0}^{k-1}l_{\mathrm{J}}(i)}],\label{eq:markovsineqforattack}
\end{align}
for $k\in\mathbb{N}$. By (\ref{eq:jammingcondition}), we have $\mathbb{E}[e^{\sum_{i=0}^{k-1}l_{\mathrm{J}}(i)}]\leq\mathbb{E}[e^{\kappa+\frac{k}{\tau}}]=e^{\kappa+\frac{k}{\tau}}$.
Therefore, it follows from (\ref{eq:markovsineqforattack}) that $\mathbb{P}[\sum_{i=0}^{k-1}l_{\mathrm{J}}(i)>\rho_{2}k]\leq\gamma_{k}^{(2)}$,
where $\gamma_{k}^{(2)}\triangleq e^{\kappa-(\rho_{2}-\frac{1}{\tau})k}$,
$k\in\mathbb{N}$. 

Now, let $\gamma_{k}\triangleq\gamma_{k}^{(1)}+\gamma_{k}^{(2)}$,
$k\in\mathbb{N}$. Using (\ref{eq:separatedprobabilityterms}), we
obtain (\ref{eq:lcond1}). Moreover, (\ref{eq:lcond2}) holds since
$\sum_{k\in\mathbb{N}}\gamma_{k}^{(1)}<\infty$, $\sum_{k\in\mathbb{N}}\gamma_{k}^{(2)}<\infty$.
\end{proof} \vskip 5pt

In comparison with Proposition~\ref{PropositionCombinedCase}, Proposition~\ref{PropositionDependentCombinedCase}
provides a more restricted range of values for $\rho$ that satisfy
Assumption~\ref{MainAssumption}. The interpretation may be that
the malicious jamming attacker is more knowledgeable and may use information
of the past random packet losses in the attack strategy.

\section{Conditions for Almost-Sure Asymptotic Stability of the Networked
Control System \label{sec:Conditions-for-Almost-Sure}}

In this section, we investigate stability of the closed-loop event-triggered
networked control system (\ref{eq:system})--(\ref{eq:controlinputatplantside}),
which is a stochastic dynamical system due to probabilistic characterization
of packet losses. Below we define almost sure asymptotic stability
for stochastic dynamical systems. 

\begin{definition}The zero solution $x(t)\equiv0$ of a stochastic
system is \emph{almost surely stable} if, for all $\epsilon>0$ and
$\bar{p}>0$, there exists $\delta=\delta(\epsilon,\bar{p})>0$ such
that if $\|x(0)\|<\delta$, then 
\begin{align}
\mathbb{P}[\sup_{t\in\mathbb{N}_{0}}\|x(t)\|>\epsilon] & <\bar{p}.
\end{align}
 Moreover, the zero solution $x(t)\equiv0$ is \emph{asymptotically
stable almost surely} if it is almost surely stable and 
\begin{align}
\mathbb{P}[\lim_{t\to\infty}\|x(t)\|=0] & =1.\label{eq:definition-convergence}
\end{align}

\end{definition}

\vskip 10pt

In Theorem~\ref{TheoremMain} below, we present sufficient conditions
for almost sure asymptotic stability of the zero solution of the dynamical
system (\ref{eq:system})--(\ref{eq:controlinputatplantside}). 

\begin{theorem} \label{TheoremMain}

Consider the linear dynamical system (\ref{eq:system}). Suppose that
the process $\{l(i)\in\{0,1\}\}_{i\in\mathbb{N}_{0}}$ characterizing
packet exchange failures in the network satisfies Assumption~\ref{MainAssumption}
with scalar $\rho\in[0,1]$. If there exist a matrix $K\in\mathbb{R}^{m\times n}$,
a positive-definite matrix $P\in\mathbb{R}^{n\times n}$, and scalars
$\beta\in(0,1),$ $\varphi\in[1,\infty)$ such that 
\begin{align}
 & \beta P-\left(A+BK\right)^{\mathrm{T}}P\left(A+BK\right)\geq0,\label{eq:betacond}\\
 & \varphi P-A^{\mathrm{T}}PA\geq0,\label{eq:varphicond}\\
 & (1-\rho)\ln\beta+\rho\ln\varphi<0,\label{eq:betaandvarphicond}
\end{align}
then the event-triggered control law (\ref{eq:attemptedpacketexchangetimes}),
(\ref{eq:controlinputatplantside}) guarantees almost sure asymptotic
stability of the zero solution $x(t)\equiv0$ of the closed-loop system
dynamics. 

\end{theorem}

\begin{proof}The proof is composed of three steps. In the initial
step, we obtain an inequality concerning the evolution of the Lyapunov-like
function $V(\cdot)$ defined by $V(x)\triangleq x^{\mathrm{T}}Px$,
$x\in\mathbb{R}^{n}$. Then, by utilizing this inequality, we will
establish almost sure stability, and then finally we show almost sure
asymptotic stability of the closed-loop system. 

First, we use (\ref{eq:system}) and (\ref{eq:controlinputatplantside})
together with $V(\cdot)$ to obtain 
\begin{align}
V(x(\tau_{i}+1)) & =x^{\mathrm{T}}(\tau_{i})\left(A+\left(1-l(i)\right)BK\right)^{\mathrm{T}}P\nonumber \\
 & \quad\,\cdot\left(A+\left(1-l(i)\right)BK\right)x(\tau_{i}),\,i\in\mathbb{N}_{0}.\label{eq:vattauplush}
\end{align}
 Now, for the case $l(i)=0$, (\ref{eq:betacond}) and (\ref{eq:vattauplush})
imply 
\begin{align}
V(x(\tau_{i}+1)) & =x^{\mathrm{T}}(\tau_{i})\left(A+BK\right)^{\mathrm{T}}P\left(A+BK\right)x(\tau_{i})\nonumber \\
 & \leq\beta x^{\mathrm{T}}(\tau_{i})Px(\tau_{i}).\label{eq:betaineq}
\end{align}
 Since $\tau_{i+1}\geq\tau_{i}+1$, it follows from (\ref{eq:attemptedpacketexchangetimes})
and (\ref{eq:betaineq}) that 
\begin{align}
V(x(t)) & \leq\beta x^{\mathrm{T}}(\tau_{i})Px(\tau_{i})\nonumber \\
 & =\beta V(x(\tau_{i})),\quad t\in\{\tau_{i}+1,\ldots,\tau_{i+1}\}.\label{eq:betaresult}
\end{align}

On the other hand, for the case $l(i)=1$, it follows from (\ref{eq:varphicond})
and (\ref{eq:vattauplush}) that 
\begin{align}
V(x(\tau_{i}+1)) & =x^{\mathrm{T}}(\tau_{i})A^{\mathrm{T}}PAx(\tau_{i})\leq\varphi x^{\mathrm{T}}(\tau_{i})Px(\tau_{i}).\label{eq:varphiineq}
\end{align}
Now if $\tau_{i+1}=\tau_{i}+1$, we have $V(x(\tau_{i+1}))\leq\varphi V(x(\tau_{i}))$
due to (\ref{eq:varphiineq}). If, on the other hand, $\tau_{i+1}>\tau_{i}+1$,
it means that $V(x(t))\leq\beta V(x(\tau_{i}))$ for $t\in\{\tau_{i}+2,\ldots,\tau_{i+1}\}$.
Therefore, since $\beta\leq\varphi$, 
\begin{align}
V(x(t)) & \leq\varphi V(x(\tau_{i})),\quad t\in\{\tau_{i}+1,\ldots,\tau_{i+1}\}.\label{eq:varphiresult}
\end{align}

Using (\ref{eq:betaresult}) and (\ref{eq:varphiresult}) we obtain
\begin{align}
V(x(\tau_{i+1})) & \leq(1-l(i))\beta V(x(\tau_{i}))+l(i)\varphi V(x(\tau_{i})),\label{eq:vineq}
\end{align}
 for $i\in\mathbb{N}_{0}$. Note that the inequality given in (\ref{eq:vineq})
characterizes an upper bound on the growth of the Lyapunov function
candidate $V(\cdot)$. 

Now, let $\eta(k)\triangleq\prod_{i=0}^{k-1}\left((1-l(i))\beta+l(i)\varphi\right)$.
It follows from (\ref{eq:vineq}) that 
\begin{align}
V(x(\tau_{k})) & \leq\eta(k)V(x(0)),\label{eq:vetaineq}
\end{align}
for $k\in\mathbb{N}$. Furthermore, since $\ln\left((1-q)\beta+q\varphi\right)=(1-q)\ln\beta+q\ln\varphi$
for $q\in\{0,1\}$, we have 
\begin{align*}
\ln\eta(k) & =\sum_{i=0}^{k-1}\ln\left((1-l(i))\beta+l(i)\varphi\right)\\
 & =\sum_{i=0}^{k-1}(1-l(i))\ln\beta+\sum_{i=0}^{k-1}l(i)\ln\varphi\\
 & =(k-L(k))\ln\beta+L(k)\ln\varphi,
\end{align*}
 where $L(k)=\sum_{i=0}^{k-1}l(i)$ by (\ref{eq:bigldefinition}).
Now since $\beta\in(0,1)$, and $\varphi\in[1,\infty)$, it follows
from Lemma~\ref{key-lemma1} that 
\begin{align*}
\limsup_{k\to\infty}\frac{\ln\eta(k)}{k} & =\limsup_{k\to\infty}\frac{1}{k}\left((k-L(k))\ln\beta+L(k)\ln\varphi\right)\\
 & \leq(1-\rho)\ln\beta+\rho\ln\varphi,
\end{align*}
 almost surely. Therefore, it follows from (\ref{eq:betaandvarphicond})
that 
\begin{align*}
\mathbb{P}[\limsup_{k\to\infty}\frac{\ln\eta(k)}{k} & <0]=1.
\end{align*}
As a consequence, $\lim_{k\to\infty}\ln\eta(k)=-\infty$, and hence,
$\lim_{k\to\infty}\eta(k)=0$, almost surely. Thus, for all $\epsilon>0$,
\begin{align*}
\lim_{j\to\infty}\mathbb{P}[\sup_{k\geq j}\eta(k) & >\epsilon^{2}]=0,
\end{align*}
and therefore, for all $\epsilon>0$ and $\bar{p}>0$, there exists
a positive integer $N(\epsilon,\bar{p})$ such that 
\begin{align}
\mathbb{P}[\sup_{k\geq j}\eta(j) & >\epsilon^{2}]<\bar{p},\quad j\geq N(\epsilon,\bar{p}).\label{eq:supetainequalit}
\end{align}

In what follows, we will employ (\ref{eq:vetaineq}) and (\ref{eq:supetainequalit})
to show almost sure stability of the closed-loop system. Note that
(\ref{eq:betaresult}), (\ref{eq:varphiresult}), and $\varphi\geq1>\beta$
imply that 
\begin{align*}
V(x(t+1)) & \leq\varphi V(x(t)),\,t\in\{\tau_{i},\ldots,\tau_{i+1}-1\},
\end{align*}
 for $i\in\mathbb{N}_{0}$. Since, $\|x\|^{2}\leq\frac{1}{\lambda_{\min}(P)}V(x)$
and $V(x)\leq\lambda_{\max}(P)\|x\|^{2}$, $x\in\mathbb{R}^{n}$,
we have 
\begin{align}
\|x(t)\|^{2} & \leq\varphi\nu\|x(\tau_{i})\|^{2},\,t\in\{\tau_{i},\ldots,\tau_{i+1}-1\},\label{eq:xtnuineq}
\end{align}
for $i\in\mathbb{N}_{0}$, where $\nu\triangleq\frac{\lambda_{\max}(P)}{\lambda_{\min}(P)}$. 

Now, let $\mathcal{T}_{k}\triangleq\{\tau_{k},\ldots,\tau_{k+1}-1\}$,
$k\in\mathbb{N}_{0}$. Then by using (\ref{eq:vetaineq}) and (\ref{eq:xtnuineq}),
we obtain $\eta(k)\geq\frac{V(x(\tau_{k}))}{V(x(0))}\geq\frac{\lambda_{\min}(P)}{\lambda_{\max}(P)}\frac{\|x(\tau_{k})\|^{2}}{\|x(0)\|^{2}}\geq\frac{1}{\nu^{2}\varphi}\frac{\|x(t)\|^{2}}{\|x(0)\|^{2}}$,
for all $t\in\mathcal{T}_{k}$, $k\in\mathbb{N}$. Hence, $\eta(k)\geq\frac{1}{\nu^{2}\varphi}\frac{\max_{t\in\mathcal{T}_{k}}\|x(t)\|^{2}}{\|x(0)\|^{2}}$,
$k\in\mathbb{N}$. By (\ref{eq:supetainequalit}), it follows that
for all $\epsilon>0$ and $\bar{p}>0$, 
\begin{align*}
 & \mathbb{P}[\sup_{k\geq j}\max_{t\in\mathcal{T}_{k}}\|x(t)\|>\epsilon\nu\sqrt{\varphi}\|x(0)\|]\\
 & \quad=\mathbb{P}[\sup_{k\geq j}\max_{t\in\mathcal{T}_{k}}\|x(t)\|^{2}>\epsilon^{2}\nu^{2}\varphi\|x(0)\|^{2}]\\
 & \quad=\mathbb{P}[\sup_{k\geq j}\frac{1}{\nu^{2}\varphi}\frac{\max_{t\in\mathcal{T}_{k}}\|x(t)\|^{2}}{\|x(0)\|^{2}}>\epsilon^{2}]\\
 & \quad\leq\mathbb{P}[\sup_{k\geq j}\eta(k)>\epsilon^{2}]\\
 & \quad<\bar{p},\quad j\geq N(\epsilon,\bar{p}).
\end{align*}
We now define $\delta_{1}\triangleq\frac{1}{\nu\sqrt{\varphi}}$.
Note that if $\|x(0)\|\leq\delta_{1}$, then (since $\nu\sqrt{\varphi}\|x(0)\|\leq1$)
we have 
\begin{align}
 & \mathbb{P}[\sup_{k\geq j}\max_{t\in\mathcal{T}_{k}}\|x(t)\|>\epsilon]\nonumber \\
 & \quad\leq\mathbb{P}[\sup_{k\geq j}\max_{t\in\mathcal{T}_{k}}\|x(t)\|>\epsilon\nu\sqrt{\varphi}\|x(0)\|]\nonumber \\
 & \quad<\bar{p},\quad j\geq N(\epsilon,\bar{p}).\label{eq:epsilon-result-part1}
\end{align}
 On the other hand, since $\varphi\geq1>\beta$, it follows from (\ref{eq:vineq})
that $V(x(\tau_{k}))\leq\varphi^{k}V(x(0))\leq\varphi^{N(\epsilon,\bar{p})-1}V(x(0))$
for all $k\in\{0,1,\ldots,N(\epsilon,\bar{p})-1\}$. Therefore, $\|x(\tau_{k})\|^{2}\leq\varphi^{N(\epsilon,\bar{p})-1}\frac{\lambda_{\max}(P)}{\lambda_{\min}(P)}\|x(0)\|^{2}=\varphi^{N(\epsilon,\bar{p})-1}\nu\|x(0)\|^{2}$.
Furthermore, as a result of (\ref{eq:xtnuineq}), 
\begin{align*}
\max_{t\in\mathcal{T}_{k}}\|x(t)\|^{2} & \leq\varphi\nu\|x(\tau_{k})\|^{2}\leq\nu^{2}\varphi^{N(\epsilon,\bar{p})}\|x(0)\|^{2},
\end{align*}
 and hence, $\max_{t\in\mathcal{T}_{k}}\|x(t)\|\leq\nu\sqrt{\varphi^{N(\epsilon,\bar{p})}}\|x(0)\|$
for all $k\in\{0,1,\ldots,N(\epsilon,\bar{p})-1\}$. Let $\delta_{2}\triangleq\epsilon\nu^{-1}\sqrt{\varphi^{-N(\epsilon,\bar{p})}}$.
Now, if $\|x(0)\|\leq\delta_{2}$, then $\max_{t\in\mathcal{T}_{k}}\|x(t)\|\leq\epsilon$,
$k\in\{0,1,\ldots,N(\epsilon,\bar{p})-1\}$, which implies 
\begin{eqnarray}
\mathbb{P}[\max_{k\in\{0,1,\ldots,N(\epsilon,\bar{p})\}}\max_{t\in\mathcal{T}_{k}}\|x(t)\|>\epsilon] & = & 0.\label{eq:epsilon-result-part2}
\end{eqnarray}
It follows from (\ref{eq:epsilon-result-part1}) and (\ref{eq:epsilon-result-part2})
that for all $\epsilon>0$, $\bar{p}>0$, 
\begin{align*}
 & \mathbb{P}[\sup_{t\in\mathbb{N}_{0}}\|x(t)\|>\epsilon]=\mathbb{P}[\sup_{k\in\mathbb{N}_{0}}\max_{t\in\mathcal{T}_{k}}\|x(t)\|>\epsilon]\\
 & \quad=\mathbb{P}[\{\max_{k\in\{0,1,\ldots,N(\epsilon,\bar{p})-1\}}\max_{t\in\mathcal{T}_{k}}\|x(t)\|>\epsilon\}\\
 & \quad\quad\quad\cup\,\{\sup_{k\geq N(\epsilon,\bar{p})}\max_{t\in\mathcal{T}_{k}}\|x(t)\|>\epsilon\}]\\
 & \quad\leq\mathbb{P}[\max_{k\in\{0,1,\ldots,N(\epsilon,\bar{p})-1\}}\max_{t\in\mathcal{T}_{k}}\|x(t)\|>\epsilon]\\
 & \quad\quad+\mathbb{P}[\sup_{k\geq N(\epsilon,\bar{p})}\max_{t\in\mathcal{T}_{k}}\|x(t)\|>\epsilon]\\
 & \quad<\bar{p},
\end{align*}
 whenever $\|x(0)\|<\delta\triangleq\min(\delta_{1},\delta_{2})$,
which implies almost sure stability. 

Now in order to establish almost sure \emph{asymptotic} stability
of the zero solution, it remains to show (\ref{eq:definition-convergence}).
To this end, note that $\mathbb{P}[\lim_{k\to\infty}\eta(\tau_{k})=0]=1$.
It follows from (\ref{eq:vetaineq}) that $\mathbb{P}[\lim_{k\to\infty}V(x(\tau_{k}))=0]=1$,
which implies (\ref{eq:definition-convergence}). Hence the zero solution
of the closed-loop system (\ref{eq:system}), (\ref{eq:attemptedpacketexchangetimes}),
(\ref{eq:controlinputatplantside}) is asymptotically stable almost
surely.\end{proof} \vskip 5pt

Theorem~\ref{TheoremMain} provides a sufficient condition under
which the event-triggered control law (\ref{eq:attemptedpacketexchangetimes}),
(\ref{eq:controlinputatplantside}) guarantees almost sure asymptotic
stability of the linear dynamical system (\ref{eq:system}) for the
case packet losses satisfy Assumption~\ref{MainAssumption}. Note
that the scalars $\beta\in(0,1)$ and $\varphi\in[1,\infty)$ in conditions
(\ref{eq:betacond}) and (\ref{eq:varphicond}) characterize upper
bounds on the growth of a Lyapunov-like function. Specifically, when
a packet exchange attempt between the plant and the controller is
successful at time $\tau_{i}$, the condition (\ref{eq:betacond})
together with (\ref{eq:attemptedpacketexchangetimes}) guarantees
that $V(x(\tau_{i+1}))\leq\beta V(x(\tau_{i}))$. On the other hand,
if a packet exchange attempt between the plant and the controller
is unsuccessful at time $\tau_{i}$, it follows from (\ref{eq:varphicond})
and (\ref{eq:attemptedpacketexchangetimes}) that $V(x(\tau_{i+1}))\leq\varphi V(x(\tau_{i}))$.
If unsuccessful packet exchange attempts are sufficiently statistically
rare (successful packet exchanges happen statistically frequently)
such that condition (\ref{eq:betaandvarphicond}) is satisfied, then
the closed-loop system stability is guaranteed. \vskip 5pt

Note that the analysis for the closed-loop system stability is technically
involved partly due to the general characterization in Assumption~\ref{MainAssumption},
which captures not only random packet losses but jamming attacks as
well. If we consider only the case of random packet losses, we may
employ methods from discrete-time Markov jump systems theory \cite{costa2004discrete}
for obtaining conditions of stability. On the other hand packet losses
due to malicious jamming attacks (Section~\ref{RemarkJammingAttacks})
cannot be described using Markov processes. Stability of a system
under jamming attacks is explored in \cite{depersis2014}, where the
analysis relies on a deterministic approach for obtaining an exponentially
decreasing upper bound for the norm of the state. In contrast, in
our analysis, after establishing that both random losses and jamming
attacks allow a probabilistic characterization, we use tools from
probability theory. Specifically, we find a stochastic upper bound
for a Lyapunov-like function and show that this stochastic upper bound
tends to zero even though it may increase at certain times.

\subsection{Feedback Gain Design for Event-Triggered Control}

In the following, we outline a numerical method for designing the
feedback gain $K\in\mathbb{R}^{m\times n}$, as well as the positive-definite
matrix $P\in\mathbb{R}^{n\times n}$ and the scalar $\beta\in(0,1)$
used in the event-triggered control law (\ref{eq:attemptedpacketexchangetimes}),
(\ref{eq:controlinputatplantside}). 

\begin{corollary} \label{Corollary} Consider the linear dynamical
system (\ref{eq:system}). Suppose that the process $\{l(i)\in\{0,1\}\}_{i\in\mathbb{N}_{0}}$
characterizing packet exchange failures in the network satisfies Assumption~\ref{MainAssumption}
with scalar $\rho\in[0,1]$. If there exist a matrix $M\in\mathbb{R}^{m\times n}$,
a positive-definite matrix $Q\in\mathbb{R}^{n\times n}$, and scalars
$\beta\in(0,1),$ $\varphi\in[1,\infty)$ such that (\ref{eq:betaandvarphicond}),
\begin{align}
\left[\begin{array}{cc}
\beta Q & \left(AQ+BM\right)^{\mathrm{T}}\\
AQ+BM & Q
\end{array}\right] & \geq0,\label{eq:corolcond1}\\
\left[\begin{array}{cc}
\varphi Q & (AQ){}^{\mathrm{T}}\\
AQ & Q
\end{array}\right] & \geq0,\label{eq:corolcond2}
\end{align}
hold, then the event-triggered control law (\ref{eq:attemptedpacketexchangetimes}),
(\ref{eq:controlinputatplantside}) with $P\triangleq Q^{-1}$ and
$K\triangleq MQ^{-1}$ guarantees almost sure asymptotic stability
of the zero solution $x(t)\equiv0$ of the closed-loop system dynamics. 

\end{corollary}

\begin{proof}Using Schur complements (see \cite{bernstein2009matrix}),
we transform (\ref{eq:corolcond1}) and (\ref{eq:corolcond2}), respectively,
into 
\begin{align}
\beta Q-\left(AQ+BM\right)^{\mathrm{T}}Q^{-1}\left(AQ+BM\right) & \geq0,\label{eq:oldcorolcond1}\\
\varphi Q-(AQ)^{\mathrm{T}}Q^{-1}AQ & \geq0.\label{eq:oldcorolcond2}
\end{align}
Now by multiplying both sides of inequalities (\ref{eq:oldcorolcond1})
and (\ref{eq:oldcorolcond2}) from left and right by $Q^{-1}$, we
obtain (\ref{eq:betacond}) and (\ref{eq:varphicond}) with $P=Q^{-1}$
and $K=MQ^{-1}$. Thus, the result follows from Theorem~\ref{TheoremMain}.
\end{proof}

\begin{figure}
\centering \includegraphics[width=0.6\columnwidth]{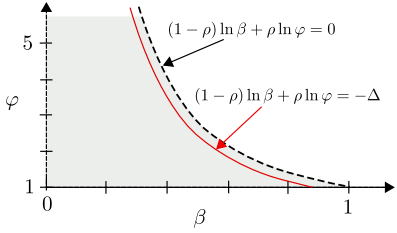}\vskip -7pt\caption{Region for $\beta\in(0,1)$ and $\varphi\in[1,\infty)$ that satisfy
(\ref{eq:betaandvarphicond}) for $\rho=0.4$ }
 \label{FigBetaandVarphi}
\end{figure}

Note that inequalities (\ref{eq:corolcond1}) and (\ref{eq:corolcond2})
are linear in $M\in\mathbb{R}^{m\times n}$ and $Q\in\mathbb{R}^{n\times n}$
for fixed $\beta\in(0,1)$ and $\varphi\in[1,\infty)$. In our method
we seek feasible solutions $M$ and $Q$ for linear matrix inequalities
(\ref{eq:corolcond1}) and (\ref{eq:corolcond2}) by iterating over
a set of values for $\beta\in(0,1)$ and $\varphi\in[1,\infty)$ that
satisfy (\ref{eq:betaandvarphicond}). We do not need to search $\beta$
and $\varphi$ in the entire space characterized by (\ref{eq:betaandvarphicond}).
We restrict the search space and only check feasibility of (\ref{eq:corolcond1})
and (\ref{eq:corolcond2}) for larger values of $\beta$ and $\varphi$
that are close to the boundary of the search space identified by $(1-\rho)\ln\beta+\rho\ln\varphi=0$.
Specifically, we set $\Delta>0$ as a small positive real number,
and then we iterate over a set of values for $\beta$ in the range
$(0,e^{-\frac{\Delta}{1-p}}]$ to look for feasible solutions $M$
and $Q$ for linear matrix inequalities (\ref{eq:corolcond1}) and
(\ref{eq:corolcond2}) with $\varphi=e^{-\frac{(1-\rho)\ln\beta+\Delta}{\rho}}$.
In this approach, feasibility of (\ref{eq:corolcond1}) and (\ref{eq:corolcond2})
is checked only for $\beta\in(0,1)$, $\varphi\in[1,\infty)$ that
are on the curve $(1-\rho)\ln\beta+\rho\ln\varphi=-\Delta$. We illustrate
the curve $(1-\rho)\ln\beta+\rho\ln\varphi=-\Delta$ with solid red
line in Fig.~\ref{FigBetaandVarphi}, where the dark shaded region
corresponds to $\beta\in(0,1)$ and $\varphi\in[1,\infty)$ that satisfy
(\ref{eq:betaandvarphicond}). Note that picking smaller values for
$\Delta>0$ moves the curve towards the boundary identified by $(1-\rho)\ln\beta+\rho\ln\varphi=0$.
Also, there is no conservatism in \emph{not} considering $\beta\in(0,1)$,
$\varphi\in[1,\infty)$ such that $(1-\rho)\ln\beta+\rho\ln\varphi<-\Delta$.
This is due to the fact that if there exist $M$ and $Q$ that satisfy
(\ref{eq:corolcond1}) and (\ref{eq:corolcond2}) for values $\beta=\tilde{\beta}$
and $\varphi=\tilde{\varphi}$, then the same $M$ and $Q$ satisfy
(\ref{eq:corolcond1}) and (\ref{eq:corolcond2}) also for larger
values $\beta>\tilde{\beta}$ and $\varphi>\tilde{\varphi}$; moreover,
for all $\tilde{\beta}\in(0,1)$, $\tilde{\varphi}\in[1,\infty)$
such that $(1-\rho)\ln\tilde{\beta}+\rho\ln\tilde{\varphi}<-\Delta$,
there exist $\beta\geq\tilde{\beta}$ and $\varphi\geq\tilde{\varphi}$
such that $(1-\rho)\ln\beta+\rho\ln\varphi=-\Delta$.

\section{Numerical Example \label{sec:Numerical-Example}}

In this section we present a numerical example to illustrate our results.
Specifically, we consider (\ref{eq:system}) with 
\begin{align*}
A\triangleq\left[\begin{array}{cc}
1 & 0.1\\
-0.5 & 1.1
\end{array}\right], & \quad B\triangleq\left[\begin{array}{c}
0.1\\
1.2
\end{array}\right].
\end{align*}
 We use the event-triggering control law (\ref{eq:attemptedpacketexchangetimes}),
(\ref{eq:controlinputatplantside}) for stabilization of (\ref{eq:system})
over a network. We consider the case where the random packet losses
in the network are characterized by the discrete-time Markov chain
$\{l_{\mathrm{R}}(i)\in\{0,1\}\}_{i\in\mathbb{N}_{0}}$ with initial
distribution $\vartheta_{0}=0$, $\vartheta_{1}=1$, and transition
probabilities $p_{0,1}(i)\triangleq0.2+0.03\sin^{2}(0.1i),$ $p_{1,1}(i)\triangleq0.2+0.03\cos^{2}(0.1i)$,
and $p_{q,0}(i)=1-p_{q,1}(i)$, $q\in\{0,1\}$, $i\in\mathbb{N}_{0}$.
Note that $\{l_{\mathrm{R}}(i)\in\{0,1\}\}_{i\in\mathbb{N}_{0}}$
satisfies (\ref{eq:p1condition}) and (\ref{eq:p0condition}) with
$p_{1}=0.23$ and $p_{0}=0.8$. Furthermore, the network is assumed
to be subject to jamming attacks characterized with $\{l_{\mathrm{J}}(i)\in\{0,1\}\}_{i\in\mathbb{N}_{0}}$
that is independent of $\{l_{\mathrm{R}}(i)\in\{0,1\}\}_{i\in\mathbb{N}_{0}}$
and satisfies (\ref{eq:jammingcondition}) with $\kappa=2$ and $\tau=5$. 

Note that $p_{1}+\frac{p_{0}}{\tau}<0.4$. It follows from Proposition~\ref{PropositionCombinedCase}
that for $\rho=0.4$, there exist $\gamma_{k}\in[0,\infty)$, $k\in\mathbb{N}$,
such that (\ref{eq:lcond1}) and (\ref{eq:lcond2}) of Assumption~\ref{MainAssumption}
hold. Furthermore, matrices 
\begin{align*}
Q & =\left[\begin{array}{cc}
0.618 & -2.119\\
-2.119 & 28.214
\end{array}\right],\,\,M=\left[\begin{array}{cc}
0.202 & -20.405\end{array}\right],
\end{align*}
 and scalars $\beta=0.55$, $\varphi=2.4516$ satisfy (\ref{eq:betaandvarphicond}),
(\ref{eq:corolcond1}), (\ref{eq:corolcond2}). Hence, it follows
from Corollary~\ref{Corollary} that the event-triggered control
law (\ref{eq:attemptedpacketexchangetimes}), (\ref{eq:controlinputatplantside})
with $P=Q^{-1}$ and $K=MQ^{-1}$, guarantees almost sure asymptotic
stabilization.

Fig.~\ref{Flo:statenorm} shows $250$ sample trajectories of the
state norm $\|x(t)\|$ obtained with the same initial condition $x_{0}=\left[1,\,1\right]^{\mathrm{T}}$
and the event-triggering mechanism parameter $\theta=1000$, but with
different sample paths for $\{l_{\mathrm{R}}(i)\in\{0,1\}\}_{i\in\mathbb{N}_{0}}$
and $\{l_{\mathrm{J}}(i)\in\{0,1\}\}_{i\in\mathbb{N}_{0}}$. Furthermore,
in Fig.~\ref{Flo:v1} we show a single sample trajectory of the Lyapunov-like
function $V(x(t))$, and in Fig.~\ref{Flo:l1} we show the corresponding
sample trajectories for $l(\cdot)$, $l_{\mathrm{R}}(\cdot)$, and
$l_{\mathrm{J}}(\cdot)$, indicating packet exchange attempt failures
due to random packet losses and jamming attacks. Note that when packet
exchange attempts fail due to a random loss or a jamming attack, the
control input is set to $0$. As a result, due to unstable dynamics
of the uncontrolled system, the Lyapunov-like function $V(\cdot)$
may grow and take a larger value at the next packet exchange attempt
instant. On the other hand, when a packet exchange attempt between
the plant and the controller is successful, the control input at the
plant side is updated. In this case $V(\cdot)$ is guaranteed to take
a smaller value at the next packet exchange attempt instant, although
it may not be monotonically decreasing. As implied by Corollary~\ref{Corollary},
the Lyapunov-like function $V(\cdot)$ eventually converges to zero
(see Fig.~\ref{Flo:v1}). 

\begin{figure}[t]
\begin{center}\includegraphics[width=0.78\columnwidth]{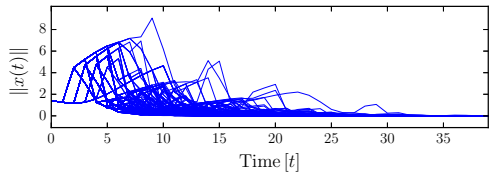}\end{center}\vskip -20pt\caption{Sample paths of the state norm}
\label{Flo:statenorm}
\end{figure}

\begin{figure}[t]
\begin{center}\includegraphics[width=0.78\columnwidth]{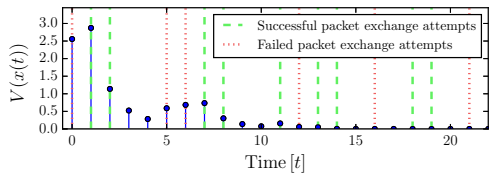}\end{center}\vskip -20pt\caption{A sample path of Lyapunov-like function $V(\cdot)$}
\label{Flo:v1}
\end{figure}

\begin{figure}[t]
\begin{center}\includegraphics[width=0.78\columnwidth]{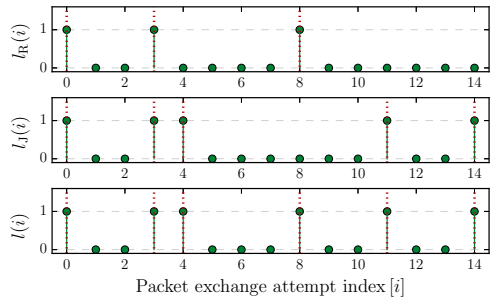}\end{center}\vskip -20pt\caption{Sample paths of $l_{\mathrm{R}}(\cdot)$, $l_{\mathrm{J}}(\cdot)$,
and $l(\cdot)$ }
\label{Flo:l1}
\end{figure}

\section{Conclusion \label{sec:Conclusion}}

In this paper, we explored event-triggered networked control of linear
dynamical systems. We described a probabilistic characterization of
the evolution of the total number of packet exchange failures in a
network that faces random packet losses and jamming attacks. Based
on this characterization, we obtained sufficient conditions for almost
sure asymptotic stabilization of the zero solution and presented a
method for finding a stabilizing feedback gain and parameters for
our proposed event-triggered control framework. In the framework that
we describe in the paper, the controller does not need the information
whether the control input packets are successfully transmitted or
lost. This type of acknowledgment messages would be useful to detect
anomalies such as jamming attacks in the network. Future extensions
include incorporating acknowledgement messages in the framework. 

\bibliographystyle{ieeetr}
\bibliography{references}

\appendix

Lemma~\ref{KeyMarkovLemma} below provides upper bounds on the tail
probabilities of sums involving a binary-valued time-inhomogeneous
Markov chain. 

\begin{aplemma}\label{KeyMarkovLemma} Let $\{\xi(i)\in\{0,1\}\}_{i\in\mathbb{N}_{0}}$
be a time-inhomogeneous Markov chain with transition probabilities
$p_{q,r}\colon\mathbb{N}_{0}\to[0,1]$, $q,r\in\{0,1\}$. Furthermore,
let $\{\chi(i)\in\{0,1\}\}_{i\in\mathbb{N}_{0}}$ be a binary-valued
process that is independent of $\{\xi(i)\in\{0,1\}\}_{i\in\mathbb{N}_{0}}$.
Assume 
\begin{align}
p_{q,1}(i) & \leq\tilde{p},\,\,q\in\{0,1\},\,\,i\in\mathbb{N}_{0},\label{eq:xicond}\\
\mathbb{P}[\sum_{i=0}^{k-1}\chi(i)\leq\tilde{c}+\tilde{w}k] & =1,\,\,k\in\mathbb{N},\label{eq:chicond}
\end{align}
where $\tilde{p}\in(0,1)$, $\tilde{w}\in(0,1]$, and $\tilde{c}\in[0,\infty)$.
It then follows that for all $\rho\in(\tilde{p}\tilde{w},\tilde{w})$,
\begin{align}
\mathbb{P}[\sum_{i=0}^{k-1}\xi(i)\chi(i)>\rho k] & \leq\psi_{k},\quad k\in\mathbb{N},\label{eq:keylemmaresult1}
\end{align}
 where $\psi_{k}\triangleq\phi^{-\rho k+1}\frac{\left((\phi-1)\tilde{p}+1\right)^{\tilde{c}+\tilde{w}k}-1}{(\phi-1)\tilde{p}}$
with $\phi\triangleq\frac{\frac{\rho}{\tilde{w}}(1-\tilde{p})}{\tilde{p}(1-\frac{\rho}{\tilde{w}})}$.
Moreover, $\sum_{k=1}^{\infty}\psi_{k}<\infty.$ 

\end{aplemma}

\vskip 10pt

In the proof of Lemma~\ref{KeyMarkovLemma}, by following the approach
used for obtaining Chernoff-type tail distribution inequalities for
sums of Bernoulli random variables (see \cite{billingsley1986}),
we use Markov's inequality. Some additional key steps (including Lemma~\ref{ExpectationLemma}
below) are also required due to the fact that in Lemma~\ref{KeyMarkovLemma}
we consider sums of (not necessarily independent) random variables
composed of the product of states of a time-inhomogeneous Markov chain
and a binary-valued process that satisfy (\ref{eq:chicond}).

\vskip 10pt

\begin{aplemma} \label{ExpectationLemma} Let $\{\xi(i)\in\{0,1\}\}_{i\in\mathbb{N}_{0}}$
be an $\mathcal{F}_{i}$-adapted binary-valued Markov chain with transition
probability functions $p_{q,r}\colon\mathbb{N}_{0}\to[0,1]$, $q,r\in\{0,1\}$.
Then for all $\phi>1$, $s\in\mathbb{N}$, and $\tilde{p}\in[0,1]$
such that 
\begin{align}
p_{q,1} & (i)\leq\tilde{p},\,\,q\in\{0,1\},\,\,i\in\mathbb{N}_{0},\label{eq:transitioncondition}
\end{align}
 we have 
\begin{align}
\mathbb{E}[\phi^{\sum_{j=1}^{s}\xi(i_{j})}] & \leq\phi\left((\phi-1)\tilde{p}+1\right)^{s-1},\label{eq:lemmaresultfors}
\end{align}
 where $i_{1},i_{2},\ldots,i_{s}\in\mathbb{N}_{0}$ denote indices
such that $0\leq i_{1}<i_{2}<\ldots<i_{s}$. \end{aplemma}

\begin{proof} The proof is based on induction. First, note that for
the case $s=1$, 
\begin{align}
\mathbb{E}[\phi^{\sum_{j=1}^{s}\xi(i_{j})}] & =\mathbb{E}[\phi^{\xi(i_{1})}]\leq\phi.\label{eq:resultfors1}
\end{align}
For the case $s=2$, the random variable $\xi(i_{1})$ is $\mathcal{F}_{i_{2}-1}$-measurable
(since $i_{1}\leq i_{2}-1$), and thus we have 
\begin{align}
\mathbb{E}[\phi^{\sum_{j=1}^{s}\xi(i_{j})}] & =\mathbb{E}[\phi^{\xi(i_{1})}\phi^{\xi(i_{2})}]\nonumber \\
 & =\mathbb{E}[\mathbb{E}[\phi^{\xi(i_{1})}\phi^{\xi(i_{2})}\mid\mathcal{F}_{i_{2}-1}]]\nonumber \\
 & =\mathbb{E}[\phi^{\xi(i_{1})}\mathbb{E}[\phi^{\xi(i_{2})}\mid\mathcal{F}_{i_{2}-1}]].\label{eq:cases2part1}
\end{align}
 Noting that $\{\xi(i)\in\{0,1\}\}_{i\in\mathbb{N}_{0}}$ is a Markov
chain, we obtain $\mathbb{E}[\phi^{\xi(i_{2})}\mid\mathcal{F}_{i_{2}-1}]=\mathbb{E}[\phi^{\xi(i_{2})}\mid\xi(i_{2}-1)]$.
Consequently, 
\begin{align}
 & \mathbb{E}[\phi^{\sum_{j=1}^{s}\xi(i_{j})}]=\mathbb{E}[\phi^{\xi(i_{1})}\mathbb{E}[\phi^{\xi(i_{2})}\mid\xi(i_{2}-1)]]\nonumber \\
 & \,\,=\mathbb{E}\Big[\phi^{\xi(i_{1})}\Big(\phi\mathbb{P}[\xi(i_{2})=1\mid\xi(i_{2}-1)]\nonumber \\
 & \,\,\quad\quad+\mathbb{P}[\xi(i_{2})=0\mid\xi(i_{2}-1)]\Big)\Big]\nonumber \\
 & \,\,=\mathbb{E}\Big[\phi^{\xi(i_{1})}\Big(\phi\mathbb{P}[\xi(i_{2})=1\mid\xi(i_{2}-1)]\nonumber \\
 & \,\,\quad\quad+1-\mathbb{P}[\xi(i_{2})=1\mid\xi(i_{2}-1)]\Big)\Big]\nonumber \\
 & \,\,=\mathbb{E}\Big[\phi^{\xi(i_{1})}\Big((\phi-1)\mathbb{P}[\xi(i_{2})=1\mid\xi(i_{2}-1)]+1\Big)\Big].\label{eq:cases2part2}
\end{align}
Then by using (\ref{eq:transitioncondition}) and (\ref{eq:resultfors1}),
we arrive at 
\begin{align}
\mathbb{E}[\phi^{\sum_{j=1}^{s}\xi(i_{j})}] & \leq\mathbb{E}\Big[\phi^{\xi(i_{1})}\Big((\phi-1)\tilde{p}+1\Big)\Big]\nonumber \\
 & =\mathbb{E}[\phi^{\xi(i_{1})}]((\phi-1)\tilde{p}+1)\nonumber \\
 & \leq\phi((\phi-1)\tilde{p}+1).\label{eq:resultfors2}
\end{align}
Hence, by (\ref{eq:resultfors1}) and (\ref{eq:resultfors2}), (\ref{eq:lemmaresultfors})
is satisfied for $s\in\{1,2\}$. 

Now, suppose that (\ref{eq:lemmaresultfors}) holds for $s=\tilde{s}>2$,
that is, 
\begin{align}
\mathbb{E}[\phi^{\sum_{j=1}^{\tilde{s}}\xi(i_{j})}] & \leq\phi\left((\phi-1)\tilde{p}+1\right)^{\tilde{s}-1}.\label{eq:resulttildes}
\end{align}
We now show that (\ref{eq:lemmaresultfors}) holds for $s=\tilde{s}+1$.
Using similar arguments that we used for obtaining (\ref{eq:cases2part1})--(\ref{eq:resultfors2}),
we obtain 
\begin{align}
\mathbb{E}[\phi^{\sum_{j=1}^{\tilde{s}+1}\xi(i_{j})}] & =\mathbb{E}[\phi^{\sum_{j=1}^{\tilde{s}}\xi(i_{j})}\phi^{\xi(i_{\tilde{s}+1})}]\nonumber \\
 & =\mathbb{E}[\mathbb{E}[\phi^{\sum_{j=1}^{\tilde{s}}\xi(i_{j})}\phi^{\xi(i_{\tilde{s}+1})}\mid\mathcal{F}_{i_{\tilde{s}+1}-1}]]\nonumber \\
 & =\mathbb{E}[\phi^{\sum_{j=1}^{\tilde{s}}\xi(i_{j})}\mathbb{E}[\phi^{\xi(i_{\tilde{s}+1})}\mid\mathcal{F}_{i_{\tilde{s}+1}-1}]]\nonumber \\
 & =\mathbb{E}[\phi^{\sum_{j=1}^{\tilde{s}}\xi(i_{j})}\mathbb{E}[\phi^{\xi(i_{\tilde{s}+1})}\mid\xi(i_{\tilde{s}+1}-1)]]\nonumber \\
 & \leq\mathbb{E}[\phi^{\sum_{j=1}^{\tilde{s}}\xi(i_{j})}]((\phi-1)\tilde{p}+1).\label{eq:casesfortildesplus1part1}
\end{align}
Using (\ref{eq:resulttildes}) and (\ref{eq:casesfortildesplus1part1}),
we arrive at 
\begin{align*}
\mathbb{E}[\phi^{\sum_{j=1}^{\tilde{s}+1}\xi(i_{j})}] & \leq\phi\left((\phi-1)\tilde{p}+1\right)^{\tilde{s}},
\end{align*}
 which completes the proof. \end{proof}

\vskip 10pt

\emph{Proof of Lemma~\ref{KeyMarkovLemma}:} First, let 
\begin{align*}
\overline{\xi}(k) & \triangleq[\xi(0),\xi(1),\ldots,\xi(k-1)]^{\mathrm{T}},\\
\overline{\chi}(k) & \triangleq[\chi(0),\chi(1),\ldots,\chi(k-1)]^{\mathrm{T}},\quad k\in\mathbb{N}.
\end{align*}
 Now let 
\begin{align*}
F_{s,k} & \triangleq\{\overline{\chi}\in\{0,1\}^{k}\colon\overline{\chi}^{\mathrm{T}}\overline{\chi}=s\},\,\,s\in\{0,1,\ldots,k\},\,k\in\mathbb{N}.
\end{align*}
 It is important to note that $F_{s_{1},k}\cap F_{s_{2},k}=\emptyset$,
$s_{1}\neq s_{2}$; moreover, due to (\ref{eq:chicond}) 
\begin{align*}
\mathbb{P}[\overline{\chi}(k)\in\cup_{s=0}^{\lfloor\tilde{c}+\tilde{w}k\rfloor}F_{s,k}] & =1,\quad k\in\mathbb{N}.
\end{align*}
It then follows that for all $\rho\in(\tilde{p}\tilde{w},1)$ and
$k\in\mathbb{N}$, 
\begin{align}
 & \mathbb{P}[\sum_{i=0}^{k-1}\xi(i)\chi(i)>\rho k]=\mathbb{P}[\overline{\xi}^{\mathrm{T}}(k)\overline{\chi}(k)>\rho k]\nonumber \\
 & \quad=\sum_{s=0}^{\lfloor\tilde{c}+\tilde{w}k\rfloor}\sum_{\overline{\chi}\in F_{s,k}}\mathbb{P}[\overline{\xi}^{\mathrm{T}}(k)\overline{\chi}(k)>\rho k\mid\overline{\chi}(k)=\overline{\chi}]\nonumber \\
 & \quad\quad\quad\cdot\mathbb{P}[\overline{\chi}(k)=\overline{\chi}].\label{eq:xibarchibarequation}
\end{align}
Due to the mutual independence of $\xi(\cdot)$ and $\chi(\cdot)$,
\begin{align}
\mathbb{P}[\overline{\xi}^{\mathrm{T}}(k)\overline{\chi}(k)>\rho k\mid\overline{\chi}(k)=\overline{\chi}] & =\mathbb{P}[\overline{\xi}^{\mathrm{T}}(k)\overline{\chi}>\rho k].\label{eq:conditionalprobabilityresolution}
\end{align}
As a result, it follows from (\ref{eq:xibarchibarequation}) and (\ref{eq:conditionalprobabilityresolution})
that for $k\in\mathbb{N}$, 
\begin{align*}
 & \mathbb{P}[\sum_{i=0}^{k-1}\xi(i)\chi(i)>\rho k]\\
 & \,\,=\sum_{s=0}^{\lfloor\tilde{c}+\tilde{w}k\rfloor}\sum_{\overline{\chi}\in F_{s,k}}\mathbb{P}[\overline{\xi}^{\mathrm{T}}(k)\overline{\chi}>\rho k]\mathbb{P}[\overline{\chi}(k)=\overline{\chi}].
\end{align*}
Next, note that $\mathbb{P}[\overline{\xi}^{\mathrm{T}}(k)\overline{\chi}>\rho k]=0$
for $\overline{\chi}\in F_{0,k}$. Hence, for all $k\in\mathbb{N}$
such that $\lfloor\tilde{c}+\tilde{w}k\rfloor=0$, we have 
\begin{align}
 & \sum_{s=0}^{\lfloor\tilde{w}k\rfloor}\sum_{\overline{\chi}\in F_{s,k}}\mathbb{P}[\overline{\xi}^{\mathrm{T}}(k)\overline{\chi}>\rho k]\mathbb{P}[\overline{\chi}(k)=\overline{\chi}]\nonumber \\
 & \quad=0.
\end{align}
 Furthermore, for all $k\in\mathbb{N}$ such that $\lfloor\tilde{c}+\tilde{w}k\rfloor\geq1$,
we have 
\begin{align}
 & \sum_{s=0}^{\lfloor\tilde{c}+\tilde{w}k\rfloor}\sum_{\overline{\chi}\in F_{s,k}}\mathbb{P}[\overline{\xi}^{\mathrm{T}}(k)\overline{\chi}>\rho k]\mathbb{P}[\overline{\chi}(k)=\overline{\chi}]\nonumber \\
 & \,=\sum_{s=1}^{\lfloor\tilde{c}+\tilde{w}k\rfloor}\sum_{\overline{\chi}\in F_{s,k}}\mathbb{P}[\overline{\xi}^{\mathrm{T}}(k)\overline{\chi}>\rho k]\mathbb{P}[\overline{\chi}(k)=\overline{\chi}].\label{eq:sumoverchibar}
\end{align}
Now, for $s\in\{1,2,\ldots,\lfloor\tilde{c}+\tilde{w}k\rfloor\}$,
let $i_{1}(\overline{\chi}),i_{2}(\overline{\chi}),\ldots,i_{s}(\overline{\chi})$
denote the indices of the nonzero entries of $\overline{\chi}\in F_{s,k}$
such that $i_{1}(\overline{\chi})<i_{2}(\overline{\chi})<\cdots<i_{s}(\overline{\chi})$.
Consequently, 
\begin{align}
\mathbb{P}[\overline{\xi}^{\mathrm{T}}(k)\overline{\chi}>\rho k] & =\mathbb{P}[\sum_{j=1}^{s}\overline{\xi}_{i_{j}(\bar{\chi})}(k)>\rho k]\nonumber \\
 & =\mathbb{P}[\sum_{j=1}^{s}\xi(i_{j}(\overline{\chi})-1)>\rho k],\label{eq:xibarequation}
\end{align}
for $\overline{\chi}\in F_{s,k}$, $s\in\{1,2,\ldots,\lfloor\tilde{c}+\tilde{w}k\rfloor\}$,
and $k\in\mathbb{N}$ such that $\lfloor\tilde{w}k\rfloor\geq1$. 

Now note that $\phi>1$, since $\rho\in(\tilde{p}\tilde{w},\tilde{w})$.
We use Markov's inequality to obtain 
\begin{align}
\mathbb{P}[\overline{\xi}^{\mathrm{T}}(k)\overline{\chi}>\rho k] & \leq\mathbb{P}[\sum_{j=1}^{s}\xi(i_{j}(\overline{\chi})-1)\geq\rho k]\nonumber \\
 & =\mathbb{P}[\phi^{\sum_{j=1}^{s}\xi(i_{j}(\overline{\chi})-1)}\geq\phi^{\rho k}]\nonumber \\
 & \leq\phi^{-\rho k}\mathbb{E}[\phi^{\sum_{j=1}^{s}\xi(i_{j}(\overline{\chi})-1)}].\label{eq:xiphiineq}
\end{align}
 Now it follows from Lemma~\ref{ExpectationLemma} that $\mathbb{E}[\phi^{\sum_{j=1}^{s}\xi(i_{j}(\overline{\chi})-1)}]\leq\phi\left((\phi-1)\tilde{p}+1\right)^{s-1}$.
Using this inequality together with (\ref{eq:sumoverchibar}) and
(\ref{eq:xiphiineq}), for all $k\in\mathbb{N}$ such that $\lfloor\tilde{c}+\tilde{w}k\rfloor\geq1$,
we obtain 
\begin{align}
 & \sum_{s=0}^{\lfloor\tilde{c}+\tilde{w}k\rfloor}\sum_{\overline{\chi}\in F_{s,k}}\mathbb{P}[\overline{\xi}^{\mathrm{T}}(k)\overline{\chi}>\rho k]\mathbb{P}[\overline{\chi}(k)=\overline{\chi}]\nonumber \\
 & \,\leq\sum_{s=1}^{\lfloor\tilde{c}+\tilde{w}k\rfloor}\sum_{\overline{\chi}\in F_{s,k}}\phi^{-\rho k}\phi\left((\phi-1)\tilde{p}+1\right)^{s-1}\mathbb{P}[\overline{\chi}_{k}=\overline{\chi}]\nonumber \\
 & \,=\phi^{-\rho k+1}\sum_{s=1}^{\lfloor\tilde{c}+\tilde{w}k\rfloor}\left((\phi-1)\tilde{p}+1\right)^{s-1}\sum_{\overline{\chi}\in F_{s,k}}\mathbb{P}[\overline{\chi}_{k}=\overline{\chi}]\nonumber \\
 & \,=\phi^{-\rho k+1}\sum_{s=1}^{\lfloor\tilde{c}+\tilde{w}k\rfloor}\left((\phi-1)\tilde{p}+1\right)^{s-1}\mathbb{P}[\overline{\chi}_{k}\in F_{s,k}]\nonumber \\
 & \,\leq\phi^{-\rho k+1}\sum_{s=1}^{\lfloor\tilde{c}+\tilde{w}k\rfloor}\left((\phi-1)\tilde{p}+1\right)^{s-1},\label{eq:finaxichiinequality}
\end{align}
where we also used the fact that $\mathbb{P}[\overline{\chi}_{k}\in F_{s,k}]\leq1$
to obtain the last inequality. Here, we have 
\begin{align}
 & \sum_{s=1}^{\lfloor\tilde{c}+\tilde{w}k\rfloor}\left((\phi-1)\tilde{p}+1\right)^{s-1}\nonumber \\
 & \quad\quad=\frac{\left((\phi-1)\tilde{p}+1\right)^{\lfloor\tilde{c}+\tilde{w}k\rfloor}-1}{\left((\phi-1)\tilde{p}+1\right)-1}\nonumber \\
 & \quad\quad\leq\frac{\left((\phi-1)\tilde{p}+1\right)^{\tilde{c}+\tilde{w}k}-1}{(\phi-1)\tilde{p}}.\label{eq:geometricseriesfinitesum}
\end{align}
Hence, (\ref{eq:finaxichiinequality}) and (\ref{eq:geometricseriesfinitesum})
imply 
\begin{align}
 & \sum_{s=0}^{\lfloor\tilde{c}+\tilde{w}k\rfloor}\sum_{\overline{\chi}\in F_{s,k}}\mathbb{P}[\overline{\xi}^{\mathrm{T}}(k)\overline{\chi}>\rho k]\mathbb{P}[\overline{\chi}(k)=\overline{\chi}]\nonumber \\
 & \quad\leq\phi^{-\rho k+1}\frac{\left((\phi-1)\tilde{p}+1\right)^{\tilde{c}+\tilde{w}k}-1}{(\phi-1)\tilde{p}},\label{eq:finalphiwinequality}
\end{align}
for all $k\in\mathbb{N}$ such that $\lfloor\tilde{c}+\tilde{w}k\rfloor\geq1$.
Note that since $\sum_{s=0}^{\lfloor\tilde{c}+\tilde{w}k\rfloor}\sum_{\overline{\chi}\in F_{s,k}}\mathbb{P}[\overline{\xi}^{\mathrm{T}}(k)\overline{\chi}>\rho k]\mathbb{P}[\overline{\chi}(k)=\overline{\chi}]=0$
for all $k\in\mathbb{N}$ such that $\lfloor\tilde{c}+\tilde{w}k\rfloor=0$,
(\ref{eq:finalphiwinequality}) holds for all $k\in\mathbb{N}$, that
is, 
\begin{align}
 & \sum_{s=0}^{\lfloor\tilde{c}+\tilde{w}k\rfloor}\sum_{\overline{\chi}\in F_{s,k}}\mathbb{P}[\overline{\xi}^{\mathrm{T}}(k)\overline{\chi}>\rho k]\mathbb{P}[\overline{\chi}(k)=\overline{\chi}]\nonumber \\
 & \quad\leq\phi^{-\rho k+1}\frac{\left((\phi-1)\tilde{p}+1\right)^{\tilde{c}+\tilde{w}k}-1}{(\phi-1)\tilde{p}},\quad k\in\mathbb{N}.\label{eq:finalphiwineequalityforallk}
\end{align}
Hence (\ref{eq:finaxichiinequality}) and (\ref{eq:geometricseriesfinitesum})
imply (\ref{eq:keylemmaresult1}). 

Our next goal is to show that $\sum_{k=1}^{\infty}\psi_{k}<\infty$.
To this end, first note that 
\begin{align}
\sum_{k=1}^{\infty}\psi_{k} & =\sum_{k=1}^{\infty}\phi^{-\rho k+1}\frac{\left((\phi-1)\tilde{p}+1\right)^{\tilde{c}+\tilde{w}k}-1}{(\phi-1)\tilde{p}}\nonumber \\
 & =\frac{\phi\left((\phi-1)\tilde{p}+1\right)^{\tilde{c}}}{(\phi-1)\tilde{p}}\sum_{k=1}^{\infty}\phi^{-\rho k}\left((\phi-1)\tilde{p}+1\right)^{\tilde{w}k}\nonumber \\
 & \quad-\frac{\phi}{(\phi-1)\tilde{p}}\sum_{k=1}^{\infty}\phi^{-\rho k}.\label{eq:psisum}
\end{align}
We will show that the series on the far right hand side of (\ref{eq:psisum})
are both convergent. First of all, since $\phi>1$, we have $\phi^{-\rho}<1$,
and thus, the geometric series $\sum_{k=1}^{\infty}\phi^{-\rho k}$
converges, that is, 
\begin{align}
\sum_{k=1}^{\infty}\phi^{-\rho k} & <\infty.\label{eq:easyconvergentsum}
\end{align}
Next, we show $\phi^{-\rho}\left((\phi-1)\tilde{p}+1\right)^{\tilde{w}}<1$.
We obtain 
\begin{align}
\phi^{-\rho}\left((\phi-1)\tilde{p}+1\right)^{\tilde{w}} & =\left(\phi^{-\frac{\rho}{\tilde{w}}}\left((\phi-1)\tilde{p}+1\right)\right)^{\tilde{w}}.\label{eq:wpower}
\end{align}
 Furthermore, 
\begin{align}
 & \phi^{-\frac{\rho}{\tilde{w}}}\left((\phi-1)\tilde{p}+1\right)\nonumber \\
 & \quad=\left(\frac{\frac{\rho}{\tilde{w}}(1-\tilde{p})}{\tilde{p}(1-\frac{\rho}{\tilde{w}})}\right)^{-\frac{\rho}{\tilde{w}}}\left(\left(\frac{\frac{\rho}{\tilde{w}}(1-\tilde{p})}{\tilde{p}(1-\frac{\rho}{\tilde{w}})}-1\right)\tilde{p}+1\right)\nonumber \\
 & \quad=\left(\frac{\tilde{p}\tilde{w}}{\rho}\right)^{\frac{\rho}{\tilde{w}}}\left(\frac{1-\tilde{p}}{1-\frac{\rho}{\tilde{w}}}\right)^{-\frac{\rho}{\tilde{w}}}\left(\frac{1-\tilde{p}}{1-\frac{\rho}{\tilde{w}}}\right)\nonumber \\
 & \quad=\left(\frac{\tilde{p}\tilde{w}}{\rho}\right)^{\frac{\rho}{\tilde{w}}}\left(\frac{1-\tilde{p}}{1-\frac{\rho}{\tilde{w}}}\right)^{1-\frac{\rho}{\tilde{w}}}.\label{eq:productofthreeterms}
\end{align}
 Note that $\frac{\tilde{p}\tilde{w}}{\rho},\frac{1-\tilde{p}}{1-\frac{\rho}{\tilde{w}}}\in(0,1)\cup(1,\infty)$.
Since $\ln v<v-1$ for any $v\in(0,1)\cup(1,\infty)$, we have 
\begin{align*}
 & \ln\left(\phi^{-\frac{\rho}{\tilde{w}}}\left((\phi-1)\tilde{p}+1\right)\right)\\
 & \quad=\frac{\rho}{\tilde{w}}\ln\left(\frac{\tilde{p}\tilde{w}}{\rho}\right)+(1-\frac{\rho}{\tilde{w}})\ln\left(\frac{1-\tilde{p}}{1-\frac{\rho}{\tilde{w}}}\right)\\
 & \quad<\frac{\rho}{\tilde{w}}\left(\frac{\tilde{p}\tilde{w}}{\rho}-1\right)+(1-\frac{\rho}{\tilde{w}})\left(\frac{1-\tilde{p}}{1-\frac{\rho}{\tilde{w}}}-1\right)\\
 & \quad=\tilde{p}-\frac{\rho}{\tilde{w}}+\frac{p}{\tilde{w}}-\tilde{p}\\
 & \quad=0,
\end{align*}
 which implies that $\phi^{-\frac{\rho}{\tilde{w}}}\left((\phi-1)\tilde{p}+1\right)<1$,
and hence by (\ref{eq:wpower}), $\phi^{-\rho}\left((\phi-1)\tilde{p}+1\right)^{\tilde{w}}<1$.
Therefore, 
\begin{align}
\sum_{k=1}^{\infty}\phi^{-\rho k}\left((\phi-1)\tilde{p}+1\right)^{\tilde{w}k}<\infty.\label{eq:difficultcasesum}
\end{align}
Finally, (\ref{eq:psisum}), (\ref{eq:easyconvergentsum}), and (\ref{eq:difficultcasesum})
imply $\sum_{k=1}^{\infty}\psi_{k}<\infty$. $\hfill \square$
\end{document}